\pdfoutput=1 
\newif\ifFull
\Fullfalse
\newif\ifAnon
\Anontrue

\documentclass{cccg26}
\usepackage{graphicx,amssymb,amsmath}

\newenvironment{restate}[1]{
 \begingroup
 \renewcommand{\thetheorem}{\ref{#1}}
 \begin{theorem}
}{
 \end{theorem}
 \addtocounter{theorem}{-1}
 \endgroup
}

\renewcommand{\emph}[1]{\textbf{\textit{#1}}}
\newcommand{\bfsw}{{\mathbf{bfsw}}}

\usepackage{booktabs}
\usepackage{mathtools}
\usepackage[capitalize, noabbrev]{cleveref}
\AddToHook{env/lemma/begin}{\crefalias{theorem}{lemma}} % temporarily tells cleveref that any label created using the theorem counter inside this environment should be treated as a lemma.

\begin{document}
\thispagestyle{empty}
\title{Layer-Respecting Linear Graph Layouts}

\author{
    Alvin Chiu\thanks{University of California, Irvine,         \texttt{chiua13@uci.edu}}
    \and
	David Eppstein\thanks{University of California, Irvine, \texttt{eppstein@uci.edu}}
    \and
	Michael T. Goodrich\thanks{University of California, Irvine, \texttt{goodrich@uci.edu}}
    \and
	Songyu (Alfred) Liu\thanks{University of California, Irvine, \texttt{songyul4@uci.edu}}
}

% \author{Anonymous}

% \author{Alvin Chiu}{University of California, Irvine}{chiua13@uci.edu}{https://orcid.org/0009-0009-6863-859X}{}
% \author{David Eppstein}{University of California, Irvine}{eppstein@uci.edu}{}{}
% \author{Michael T. Goodrich}{University of California, Irvine}{goodrich@uci.edu}{https://orcid.org/0000-0002-8943-191X}{}
% \author{Songyu (Alfred) Liu}{University of California, Irvine}{songyul4@uci.edu}{https://orcid.org/0009-0003-1255-7156}{}

%\authorrunning{A. ~Chiu, D. ~Eppstein, M.T.~Goodrich, and S.~Liu}

% \Copyright{Alvin Chiu, David Eppstein, Michael T. Goodrich, and Songyu (Alfred) Liu} 

% \ccsdesc[500]{Theory of computation}

% \keywords{graph drawing, layered graphs, arc diagrams, cylindric drawings, FPT algorithms} 

%\category{Short (Theory)}

\maketitle

\begin{abstract}
We show how to visualize a graph, $G=(V,E)$, as a layered drawing, layer-respecting arc diagram, or layer-respecting linear cylindric drawing with a minimum number of edge crossings, where layer-respecting means that layers appear in order on a single line and vertices are grouped by their layers. Even though this problem is NP-hard for general arc diagrams, we show how to create such diagrams with fixed-parameter tractable linear-time algorithms, where the parameter that allows this is the width of a layered graph. Such a layered graph can be obtained from a breadth-first search (BFS), in which case the width is upper bounded by a graph width parameter called the BFS width.
\end{abstract}

\section{Introduction}
A \emph{layering} of a graph $G=(V,E)$ 
is a partition, $\{L_1,L_2,\ldots, L_H\}$, of $V$ into disjoint subsets, 
which are referred to
as \emph{layers}, such that $\bigcup_{i=1}^H L_i=V$ and
for every edge $(v,w)\in E$, if $v\in L_i$ and $w\in L_j$,
then $|i-j|\le 1$; see, e.g.,~\cite{layers,bannister2019track}. The \emph{width} of a layered graph is the number of vertices in its largest layer \cite{healy_how_2002}. One nice property of layerings of a connected graph is that removing all the vertices in any layer that is not the topmost or bottommost layer will disconnect the graph. The pioneering work on graph 
separators by Lipton and Tarjan~\cite{lipton1979separator}
used a layering of a connected graph $G$ computed from a breadth-first
search (BFS) on an arbitrarily chosen vertex in $G$ by defining each layer by its distance from the starting vertex. We define such a layering formed using the BFS tree from some vertex as a \emph{BFS layering}.
Note that a graph layering is not necessarily a
BFS layering, but every BFS layering is a graph layering.
To see the latter, suppose we have a BFS layering with starting vertex $r$. $V$ is naturally partitioned into disjoint subsets by distances from $r$. For any edge $(v,w)\in E$ where $v\in L_i \text{ and } w\in L_j$, we have $|i-j|\le 1$ by the reverse triangle inequality for the distances between $v$, $w$, and $r$.
See Figure~\ref{fig:sugiyama}.

\begin{figure}[hbt]
    \centering
    \includegraphics[width=0.5\linewidth]{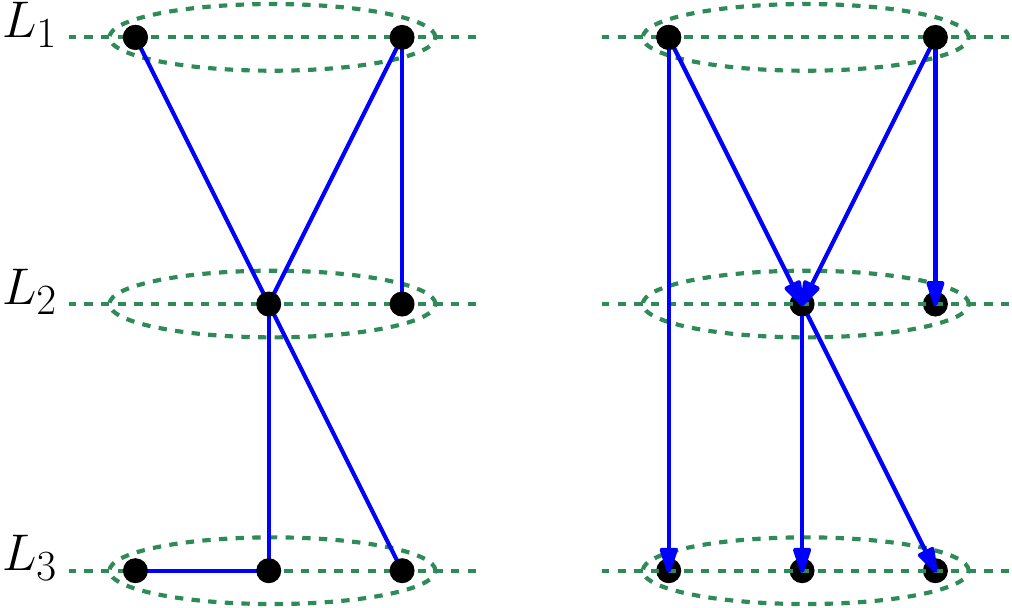}
    \caption{Left: a layered drawing of a layered graph. The partition indicated by $L_i$  is  a valid layering but not a BFS layering.
    Right: a Sugiyama-style drawing of a directed graph. The edge between the topmost vertex and bottommost vertex spans multiple layers. Thus, the partition indicated by $L_i$ is not a valid layering for the graph on the right.}
    \label{fig:sugiyama}
\end{figure}
The work on layered drawings of graphs by
Sugiyama, Tagawa, and Toda~\cite{sugiyama} pioneered the graph drawing paradigm now known
as a ``Sugiyama-style'' drawing. This style is for directed graphs, and
in this style, vertices are drawn as points on horizontal lines and each line is one layer; all edges point downwards; straight-line segments are used for edges joining two points in consecutive layers and polygonal chains may be used for edges joining two points spanning multiple layers. See, e.g.,~\cite{eppstein2007confluent}.
Strictly speaking, a Sugiyama-style drawing does not allow for edges
between vertices in the same layer, but we do not make this restriction
in this paper. See Figure \ref{fig:sugiyama}.
In a Sugiyama-style drawing, the number of layers is lower bounded by the number of vertices on a longest path, since each vertex on this path must be on a different layer \cite{ruegg_generalization_2016}. With any BFS layering, the number of layers is upper bounded by the number of vertices on a longest path.
% Sugiyama-style drawings are especially effective when
% the number of layers is small and the size of each layer is large,
% relative to the number of vertices,\footnote{In this paper,
%    we generally use $n$ to denote the number of vertices in a 
%    given graph, $G$, and $m$ to denote the number of edges in $G$.} $n$,
% but it is not as effective if the number of layers is large and the
% size of each layer is small.
Thus, we are interested in 
an alternate graph drawing paradigm.
% which is better suited for drawing layered graphs where there are
% a large number of layers and layers are small. 
Furthermore, we also study \emph{linear graph layouts}  \cite{bekos_online_2023, auer_plane_2011}, in particular arc diagrams and linear cylindric drawings. As it turns out, under certain constraints, crossing minimization for these different styles will share the same dynamic programming framework.

\paragraph*{Arc Diagrams}
One type of linear graph layout is
an \emph{arc diagram} of a graph, 
which is a style of graph drawing 
where the vertices of $G$ 
are points on a horizontal line, and the edges of $G$ are drawn as semicircular arcs or straight-line segments 
for edges joining pairs of consecutive points. Applications for arc diagrams include network 
visualization~\cite{burch2021dynamic,debiasi,komarek_network_2015}, as well as string structure visualization~\cite{wattenberg_arc_2002}.  
We focus our work on \emph{layer-respecting} arc diagrams in which the vertices of each layer of the input graph  appear consecutively within the linear ordering of vertices on the horizontal line and layers appear in order. See Figure \ref{fig:layer-respecting}. A formal definition is given in \cref{sec:preliminaries}. 
For more general linear graph layouts and Sugiyama-style drawings, ordering the vertices to minimize crossings can be a hard problem. For Sugiyama-style drawings, people have studied various crossing minimization heuristics. The most common technique is to consider two consecutive layers at a time \cite{bachmaier_crossing_2010}. Once the vertex ordering of one layer is fixed, the ordering of the vertices in the other layer determines the number of crossings. This ordering can be chosen by the barycenter heuristic \cite{sugiyama} or the sifting heuristic \cite{bachmaier_crossing_2010}.
However, we show that this crossing-minimization problem can be solved optimally in layer-respecting layouts, when the width of the layering is small.

\begin{figure}[hbt]
    \centering
    \includegraphics[width=\linewidth]{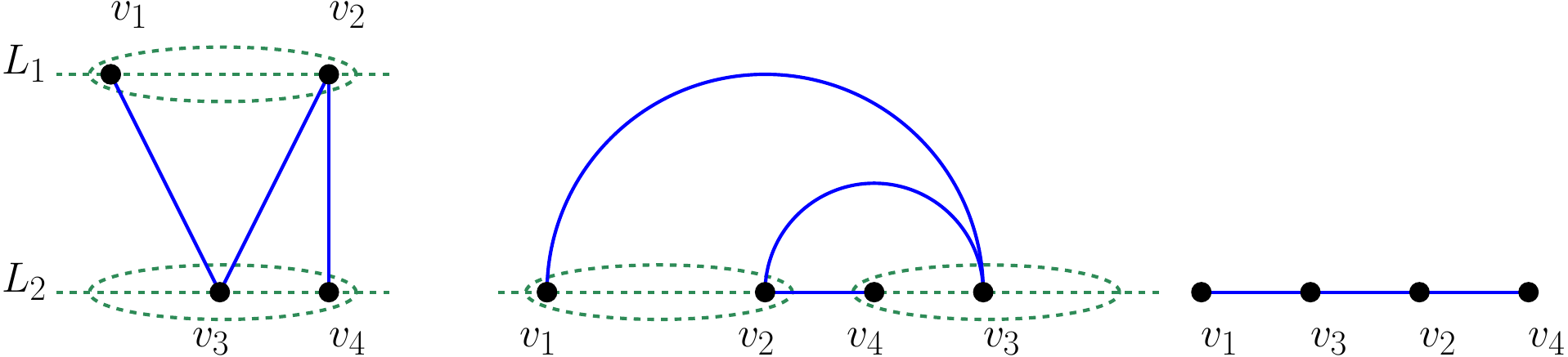}
    \caption{Left: a layered drawing of a layered graph.
    Middle: a layer-respecting arc diagram of this graph. Vertices of each layer of the input graph appear consecutively even though $v_4$ is now on the left of $v_3$. 
    Right: an arc diagram of this graph that is not layer-respecting. Vertices in $L_1$ of the input graph do not appear consecutively since $v_3$ is in between $v_1$ and $v_2$.
    }
    \label{fig:layer-respecting}
\end{figure}

To be more precise, arc diagrams may have edges drawn as curves (such as piecewise semicircular curves) that cross the horizontal line on which the vertices are arranged. For example, arcs composed of two semicircles that cross the horizontal line are called \emph{biarcs}~\cite{chaplick_monotone_2024}.
See Figure~\ref{fig:arc}. A \emph{proper arc} is a curve that does not cross this horizontal line, and a \emph{proper arc diagram} is an arc diagram with only proper arcs \cite{chaplick_monotone_2024}. 

\begin{figure}[hbt]
    \centering
    \includegraphics[width=0.5\linewidth]{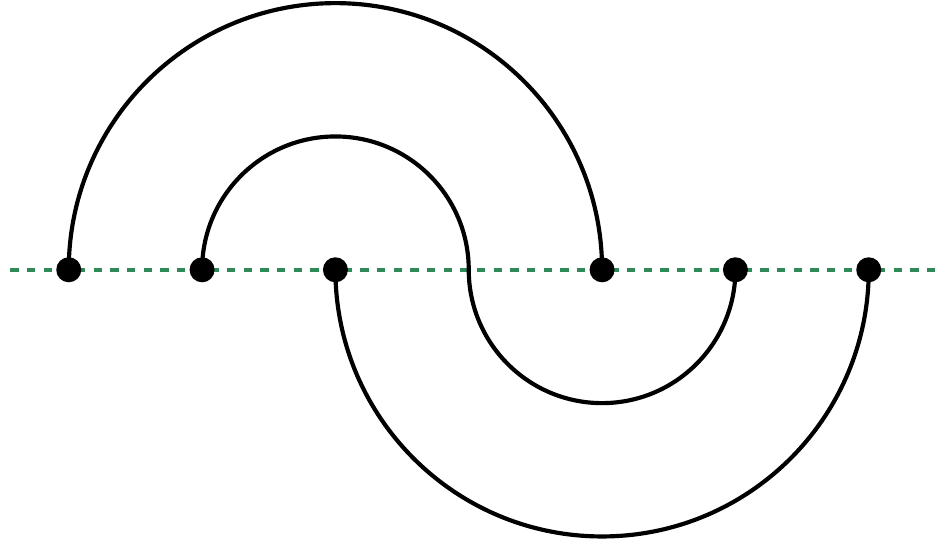}
        \caption{An example arc diagram with a biarc. 
    Notice that if we turn this biarc into a proper arc, we will introduce one crossing.
    }
    % \caption{Left: an example arc diagram with a biarc. 
    % Right: a 1-page proper arc diagram of the same graph. The definition of pages is in Section \ref{sec:pages}.
    % }
    \label{fig:arc}
\end{figure}

Our arc diagrams draw proper arcs as semicircles or as straight-line segments between consecutive points.
% but in some cases we  also use quarter-circles in order to provide room for additional edges to reach vertices without crossing the proper arcs that are incident to the same vertices. In this case, the tangent of the arc no longer forms a perpendicular angle to the horizontal line, but rather a 45-degree angle. See Section \ref{paragraph:crossings} and Figure \ref{fig:5-styles}.
It turns out that whether we allow non-proper arcs or not in an arc diagram has non-trivial complexity implications, as a maximal planar graph has a proper arc diagram without crossings if and only if it contains a Hamiltonian cycle \cite{bernhart_book_1979}, and testing whether a given graph has a crossing-free proper arc diagram is NP-complete~\cite{vlsi}.
Nevertheless, following in the tradition in graph drawing, we focus on proper arc diagrams, or simply ``arc diagrams'' when the context is clear \cite{klawitter_experimental_2018,bannister_crossing_2018}. In this paper, we are interested in layer-respecting arc diagrams, where the vertices in each layer have consecutive positions
in the linear ordering.

\paragraph*{Linear Cylindric Drawings}
Beyond arc diagrams, we also consider another drawing style for linear graph layouts,
\emph{linear cylindric drawings}~\cite{auer_plane_2011}. 
These drawings are made on the surface of a cylinder, with vertices placed on a line $L$  parallel to the axis of a cylinder. Edges are drawn as curves, monotone with respect to the direction of the cylinder's axis, that do not cross $L$. These can be either proper arcs or curves that wrap once around the cylinder.
See Figure \ref{fig:cylinder-spiral}. 

\begin{figure}[b]
    \centering
    \includegraphics[width=.9\linewidth]{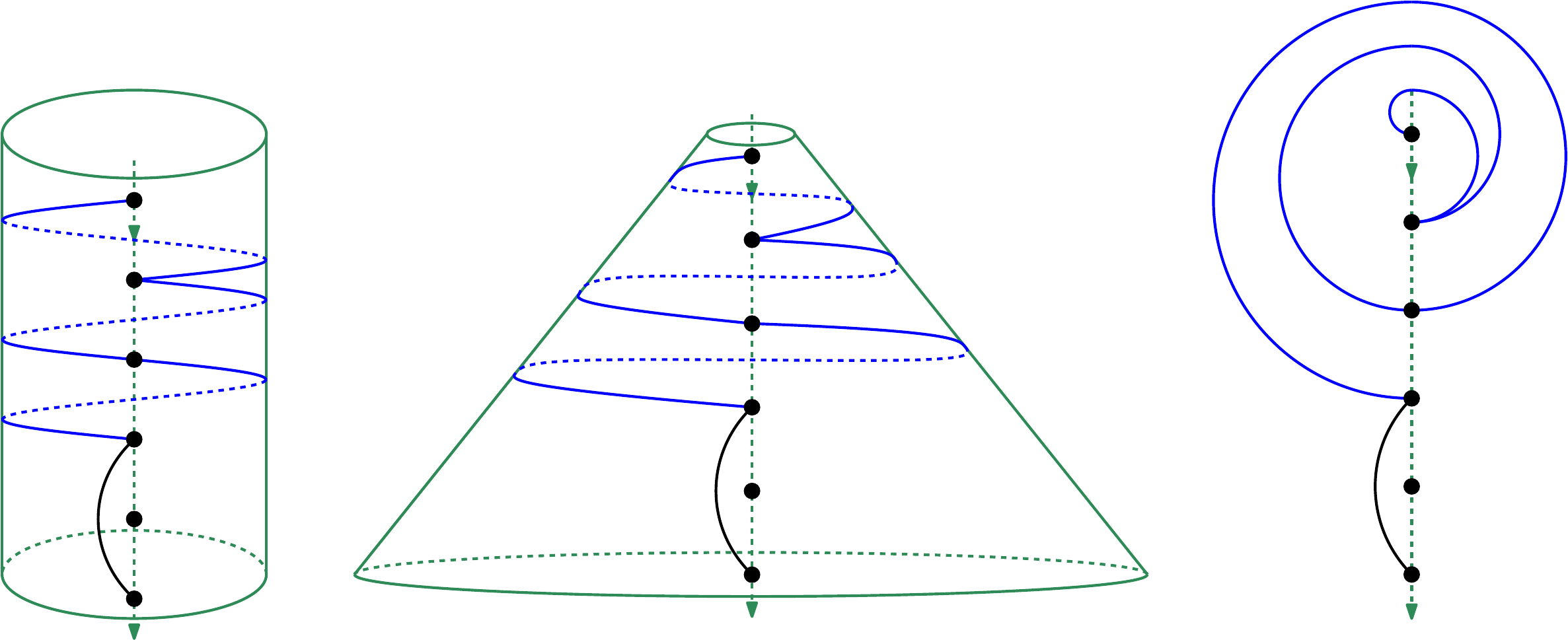}
    \caption{A linear cylindric drawing (left) transformed into a linear spiral drawing (right). We stretch one end of the cylinder in a topology-preserving manner, transforming the cylinder into an conical frustum during the process (middle).
    Arcs on the cylinder's surface (shown in black) will be preserved. The rest of the edges that wrap around the cylinder will become spirals in the linear spiral drawing. Note that vertices are on the front line of the cylinder, not its axis.}
\label{fig:cylinder-spiral}
\end{figure}

Non-crossing cylindric drawings correspond to a class of graphs called deque graphs  named after the double-ended queue data structure (a deque) and the formal definition can be found in \cite{auer_plane_2011}.
% To see this correspondence, suppose we have a linear graph layout. Scan the vertices from left to right. When we encounter the left endpoint of an edge, we add it to a deque, and we remove it from the deque when we encounter its right endpoint. If all edges can be processed in this way without violating the deque property, then we have a \emph{deque layout} and the graph is a \emph{deque graph}. For a general graph, a deque layout is not necessarily possible, and we allow crossings. 
To the best of our knowledge, minimizing the number of crossings over all linear cylindric drawings has not been studied before. This problem appears difficult since cylindric drawings contain two arc diagrams.

Linear cylindric drawings may be flattened into plane drawings for easier visualization, preserving the number of crossings, in at least two ways. In the first method, we stretch one end of the cylinder in a topology-preserving manner, keeping the other end fixed, flattening the cylinder into an annulus and producing what we call a \emph{linear spiral drawing}. (See Figure \ref{fig:cylinder-spiral} again.) 
In the second method, which is not topology preserving, we cut the cylinder along the line of the vertices, producing two copies of each vertex on the two cut copies of this line, and flattening the surface by unrolling the cylinder~\cite{auer_plane_2011}. This is called an 
\emph{unrolled cylinder} and Figure \ref{fig:5-styles} shows an unrolled cylinder that is also layer-respecting. 
Since the number of crossings will be preserved using both methods \cite{auer_plane_2011}, we consider unrolled cylinders when minimizing the number of crossings.
As we show in this paper, minimizing the number of crossings in a layer-respecting
linear cylindric drawing is similar in spirit to minimizing crossings in a layer-respecting arc diagram. 

% \begin{figure}[hbt]
%     \centering
%     \includegraphics[width=.9\linewidth, trim=15in 0in 0in 13in, clip]{graphics/5-styles.pdf}
%     \caption{A layer-respecting linear cylindric drawing.
%     Dashed ellipses indicate layers. The leftmost vertex is in its own layer.}
%     \label{fig:1-style}
% \end{figure}

\subsection{Related Prior Work} \label{sec:pages}
Arc diagrams are related to \emph{book drawings}, where the vertices lie on a horizontal line, and each edge is drawn on one of several half-planes bounded by the line, called \emph{pages}~\cite{masuda_crossing_1990, shahrokhi_book_1996}. An arc diagram with edges above and below the horizontal line can be seen as a 2-page book drawing\footnote{For convenience, straight-line segments between consecutive points are considered to be arcs above the line.}, where the half-plane above the line is considered page one and the one below is considered page two. Similarly, an arc diagram with all edges above the horizontal line is a 1-page book drawing. 
A \emph{book embedding} is a book drawing with zero edge crossings.
% , and the \emph{book thickness} of a graph is then the smallest number of half-planes (or ``pages'') needed to obtain a book embedding of the graph; see, e.g.,~\cite{bernhart_book_1979}.
The  \emph{$k$-page book crossing number} of a graph is the minimum number of edge crossings in any $k$-page book drawing of it. This minimum can be approximated within a multiplicative factor of $O(\log^2n)$ in polynomial time under certain conditions, where $n$ is the number of vertices~\cite{shahrokhi_book_1996}.

In prior work, Bannister and Eppstein~\cite{bannister_crossing_2018} study
the parameterized complexity of computing the 1-page and 2-page crossing number for graphs of low treewidth. However, there are important differences between the work of Bannister and Eppstein and our work. Their work uses Courcelle’s theorem and only proves the existence of fixed-parameter tractable algorithms, while we directly provide dynamic programming algorithms. Their complexity result for 2-page drawings is parameterized by the sum of the crossing number and treewidth, whereas our new results are parameterized only by  the \emph{width} of a layered graph. Thus, their approach with treewidth can only handle graphs with very small crossing numbers, while our width parameter can handle graphs with arbitrarily large crossing numbers.
In independent work, \cite{hegemann_chunky_2026} study a dynamic programming algorithm similar to ours. They propose a fixed-parameter tractable algorithm parameterized by the capacity to minimize crossings of chunky chains. In a chunky chain, vertices are partitioned into buckets, and short edges are edges with endpoints in the same bucket or adjacent buckets. Crossings are only between short edges by design. Capacity is a parameter that upper bounds the number of vertices in any bucket. 

% If a fixed ordering of the vertices along the horizontal line is given, then we have the corresponding definitions for \emph{fixed book thickness} and \emph{fixed book crossing number}~\cite{masuda_crossing_1990}. 
% The fixed book thickness problem is closely related to the coloring problem.
% Given a graph $G$ and a fixed ordering, we can construct a conflict graph $H$, where every vertex of $H$ corresponds to an edge of $G$ and two vertices of $H$ are neighbors if their corresponding edges cross each other in $G$  \cite{agrawal_eliminating_2024}. If edges of $G$ are assigned to $k$ pages, it is equivalent to assigning $k$ colors to the vertices of $H$. If two neighboring vertices in $H$ have different colors, then their corresponding edges in $G$ are assigned to different pages and no longer cross. Thus, computing the  fixed book thickness of $G$ is as hard as computing the chromatic number of $H$ \cite{agrawal_eliminating_2024}. The hardness of this coloring problem is discussed in \cite{bachmann_3-coloring_2023}.
If a fixed ordering of the vertices along the horizontal line is given, then we have the corresponding definition for \emph{fixed book crossing number}~\cite{masuda_crossing_1990}. 
A $k$-page book drawing with the minimum number of crossings can be found in $2^m  n^{O(1)}$ time\footnote{In this paper,
   we generally use $n$ to denote the number of vertices in a 
   given graph $G$, and $m$ to denote the number of edges in $G$. Note that one can assume that the number of pages $k$ is less than the number of edges $m$, because if $k \ge m$, we can trivially achieve zero crossings.}, for a fixed ordering and positive integer $k$ \cite{agrawal_eliminating_2024}. However, complexity issues remain, as it is still NP-hard to compute the 2-page fixed book crossing number, even for a fixed ordering of the vertices~\cite{masuda_crossing_1990}.
As a result, finding an arc diagram of a graph with the minimum number of edge crossings is NP-hard in general, even when the vertex ordering is fixed~\cite{masuda_crossing_1990}.

% Deque graphs are named after the data structure deque. Apart from deques, other data structures can be used \cite{heath_laying_1992, bekos_online_2023}. For example, if we use a stack, we have the definition of stack layouts and stack graphs. Multiple stacks can be used, in which case each edge is pushed to one of them when its left endpoint is processed. Such stack layouts are equivalent to book embeddings where each page corresponds to one stack \cite{heath_laying_1992}.
\begin{table*}
     \caption{The five drawing styles we consider. Here, $w$ is the width of the layered graph.}
\centering 
    \begin{tabular}{l c c c c}
    \toprule
    & \multicolumn{2}{c}{Arcs above the line} & \multicolumn{2}{c}{Arcs above/below} \\
    \cmidrule(r){2-3} \cmidrule(l){4-5}
    Drawing style & Style & Complexity & Style & Complexity \\
    \midrule
    Layered drawing &  1 & $O(w^4 (w!)^2 n)$ &  3 & $O(w^4 (w!)^2 2^{3 w^2}n)$  \\
    Layer-respecting arc diagram &  2 & $O(w^4 (w!)^2 n)$ &  4 & $O(w^4 (w!)^2 2^{3 w^2}n)$ \\
    Layer-respecting linear cylindric drawing & \multicolumn{2}{c}{Not applicable} &  5 & $O(w^4 (w!)^2 4^{3 w^2}n)$\\
    \bottomrule
    \end{tabular}
    \label{tab:5-styles}
\end{table*}

\begin{figure*}
\centering
    \includegraphics[width=.9\linewidth]{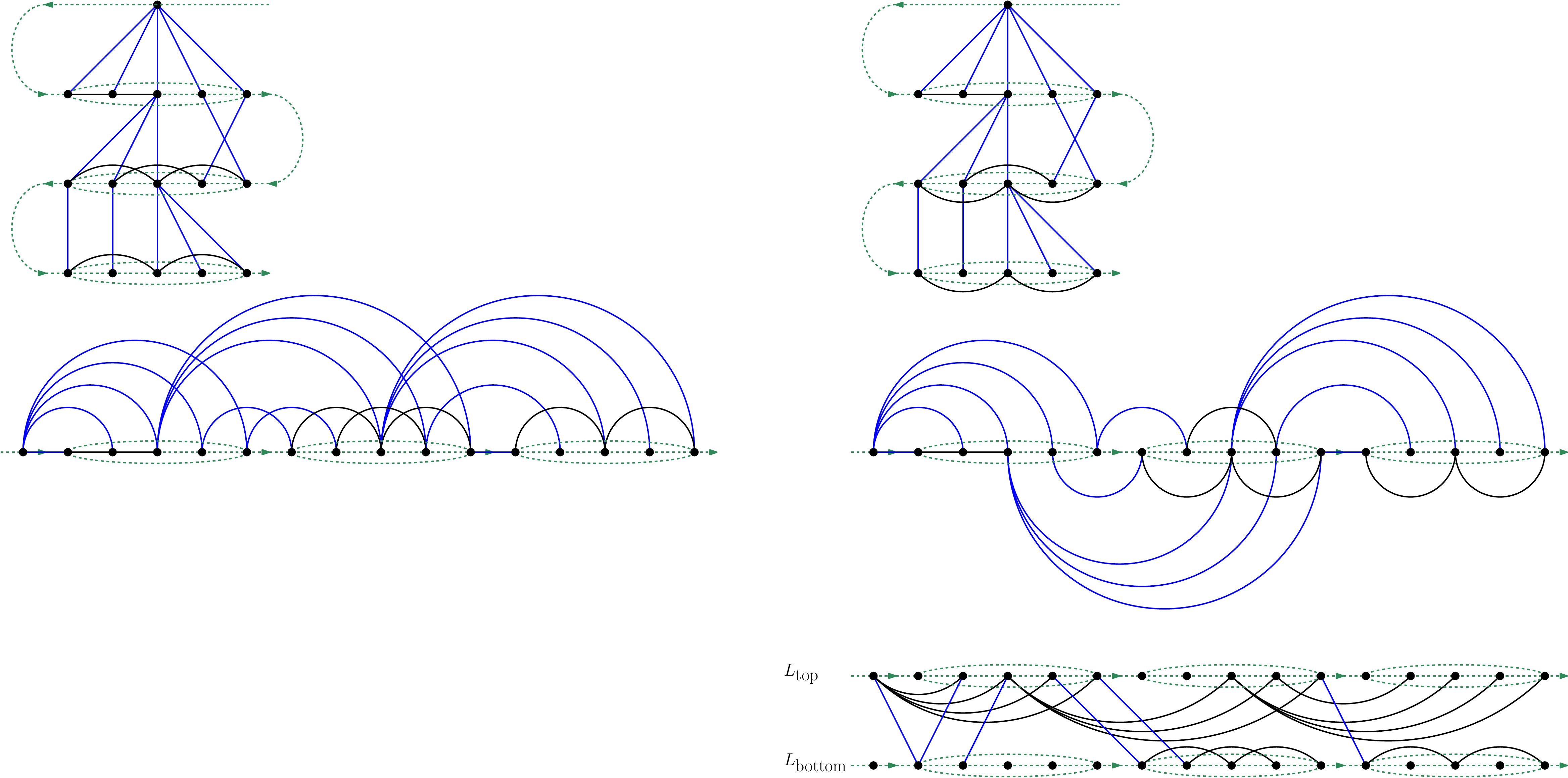}
    \caption{The five drawing styles we consider. 
   Left: styles 1 and 2 in column 1.  Right: styles 3, 4, and 5 in column 2. All the drawings correspond to the same graph. }
    \label{fig:5-styles}
    \vspace*{-12pt}
\end{figure*}

\subsection{Our Results}
The minimum width of a graph over all its possible layerings is within a constant factor of the graph's bandwidth \cite{dubey_hardness_2011} and the best known polynomial time algorithm only provides a $O(\log^3 n \sqrt{\log \log n})$-approximation for bandwidth on general graphs with high probability \cite{dunagan_euclidean_2001}. 
Thus, there is no efficient way to find a layering that achieves minimum width. In this paper we focus on graphs that either already have
a layering with small width given as part of the input, or graphs with small BFS width, which is the maximum width over all
possible BFS layerings for a graph~\cite{eppstein_et_al:LIPIcs.ESA}. Interestingly, Eppstein, Goodrich,
and Liu~\cite{eppstein_et_al:LIPIcs.ESA} show that this parameter is related to other graph parameters,
including the bandwidth and pathwidth of a graph.
In an arc diagram where vertices are at consecutive integer coordinates, the height of the arc diagram is upper bounded by the bandwidth of the vertex ordering \cite{eppstein_et_al:LIPIcs.ESA}. In both of the cases that we consider in this paper, the input graph's layering directly provides a vertex ordering with small bandwidth \cite{dubey_hardness_2011}, so its arc diagram will fit in a wide and short bounding box. The layer-respecting constraint further allows us to minimize the number of crossings.

For layered graphs, the layering we use can either
be given in advance or we can compute a BFS layering for the input graph.
In either case, we are interested in \emph{layer-respecting} linear graph layouts, where the vertices in each
layer are placed at consecutive positions as a group, so as to still allow visualization of the layering (e.g., with colors
and/or dashed grouping ovals as shown in Figure \ref{fig:5-styles}).
However, this drawing paradigm does not prescribe the ordering of the vertices in each layer; hence, to optimize visualization aesthetics we desire orderings that minimize the number of edge crossings.

In this paper, we show that when the graph has bounded width (in the case of a prescribed layering) or bounded BFS width~\cite{eppstein_et_al:LIPIcs.ESA} (if the graph does not come with a given layering), then we can minimize the number of crossings in a layer-respecting arc diagram or layer-respecting
linear cylindric drawing with fixed-parameter tractable time complexity. Furthermore, our framework can construct a layered drawing with the same time complexity; the constraints for which are described in Section \ref{sec:preliminaries}. 
% Our techniques also apply to general graphs with no edges skipping layers, namely, the two endpoints of any edge are in the same layer or consecutive layers. We use BFS layers as a concrete example because they are easy to obtain.
%
We consider five different drawing styles; see Figure \ref{fig:5-styles} and \cref{tab:5-styles}.
% Our specific results are shown in Table \ref{tab:results}. 

%%\mike{Can we also include a figure showing the 5 drawing styles?}
%\alfred{The two column layout matches the table below. It is a bit small and maybe we can stretch it vertically}

% \begin{table*}
%      \caption{The five drawing styles we consider. Here, $w$ is the width of the layered graph.      }
% \centering 
%     \begin{tabular}{lcc}
%     \toprule
%     &  Arcs above the line & Arcs above/below  \\
%     \midrule
%     Layered drawing & style 1 $O(w^4 (w!)^2 n)$ & style 3 $O(w^4 (w!)^2 2^{3 w^2}n)$  \\
%     Layer-respecting     arc diagram & style 2 $O(w^4 (w!)^2 n)$ & style 4 $O(w^4 (w!)^2 2^{3 w^2}n)$ \\
%     Layer-respecting     linear  cylindric drawing & not applicable & style 5 $O(w^4 (w!)^2 4^{3 w^2}n)$\\
%     \bottomrule
%     \end{tabular}
%     \label{tab:5-styles}
% \end{table*}

% \begin{table}
%      \caption{Our results. Here, $w$ is the width of the layered graph. 
%      }
% \centering   
%     \begin{tabular}{lc}
%     \toprule
%     & Time complexity \\
%     \midrule
%         Styles 1 and 2 &  $O(w^4 (w!)^2 n)$   \\
%     Styles 3 and 4 & $O(w^4 (w!)^2 2^{3 w^2}n)$  \\
%  Style 5 & $O(w^4 (w!)^2 4^{3 w^2}n)$ \\

%     \bottomrule
%     \end{tabular}
% \label{tab:results}
% \end{table}

\section{Preliminaries} \label{sec:preliminaries}
\subsection{BFS Width}
\emph{BFS width} is a graph width parameter introduced by Eppstein, Goodrich, and Liu~\cite{eppstein_et_al:LIPIcs.ESA}, whose definition we adapt to match the graph drawing literature \cite{battista_graph_1998}.
% Let $G=(V,E)$ be an undirected unweighted graph, and let a layering of $G$ be $\{L_1,L_2,\ldots, L_H\}$. 
% the \emph{width} of $G$ is $\max_{i=1}^H |L_i|$ and the \emph{height} of $G$ is the number of layers, $H$.
Given a layering of $G$ $\{L_1,L_2,\ldots, L_H\}$, the \emph{span} of an edge  $(u_1, u_2)$ with $u_1 \in L_i$ and  $u_2 \in L_j$ is $|i-j|$ \cite{battista_graph_1998}. Recall that by the definition of layering, edges are between consecutive layers or within the same layer, so the span of any edge is at most 1 for layered graphs.

We  obtain a layering either by receiving a layering as a part of the input or by performing a breadth-first search of $G$ from an arbitrary starting vertex, $v$, where the edges traversed define a breadth-first search tree $T$ such that the depth of each
vertex in $T$ equals its distance from $v$, the number of edges in a shortest path from $v$.
We define the \emph{BFS layer} $i$, denoted as $L_i$, to consist of all vertices at distance $i - 1$ from $v$; if $u \in L_i$, we say that $u$ has \emph{layer number} $i$. 
Then we can finally define the \emph{BFS width} of $G$, denoted by $\bfsw(G)$, as the maximum width, taken over all possible starting vertices $v\in V$ for a BFS layering.
Since edges are between consecutive layers or within the same layer, any vertex can only be connected to vertices in at most three layers excluding itself.
\begin{lemma}
    Graph $G$ has maximum degree $\Delta$ that is at most $3 \bfsw(G) - 1$. The number of edges satisfies
    $ |E| \le {(3 \bfsw(G) - 1) |V|} /   {2}$.
\end{lemma}
We can then express the $O(|V|+|E|)$ time complexity of BFS in terms of the BFS width.
\begin{lemma}   \label{lemma:bfs}
    % The time complexity of BFS is $O({(3 \bfsw(G) + 1) |V|} /   {2})$.
     The time complexity of BFS is $O(\bfsw(G) \cdot |V|)$.
\end{lemma}

\subsection{Layer-Respecting Linear Graph Layouts}   
%\mike{Please rewrite this section and the rest of the paper to allow for a layering to be given as a part of the input.}
Suppose we are given a graph $G=(V,E)$ and a layering $\{L_1,L_2,\ldots, L_H\}$ of $G$.
% \begin{definition}
    A \emph{linear layout} is a bijective function $\sigma : V \to \{1, \dots, |V|\}$ that specifies the positions of the vertices from left to right \cite{auer_plane_2011}.
        A \emph{linear graph layout} of $G$ is a drawing in which the horizontal positions of the vertices are specified by $\sigma$.
% \end{definition}
 The function $\sigma$ determines how $V$ is arranged in a linear order, which does not depend on the layering in general, namely the partition of $V$.
In this paper, a linear graph layout refers to either an arc diagram or a linear cylindric drawing. In a linear cylindric drawing, the top and bottom lines contain the same vertices in the same order, hence the horizontal positions of vertices can be specified by a single function. Now we add a constraint on a linear graph layout so that it ``respects'' the layering.
% \begin{definition}\label{def:layer-respecting}
    A linear graph layout is \emph{layer-respecting} if for all vertices $u_1 \in L_i$ and  $u_2 \in L_j$ where $i < j$, we have $\sigma(u_1) < \sigma(u_2)$.  
% \end{definition}
While the vertices are arranged left-to-right by their layer numbers, vertices with the same layer number can be in any order; hence, we are free to rearrange vertices within each layer, e.g., to minimize the number of edge crossings.

\subsection{Layered Drawings}  
In addition to the layer-respecting linear graph layouts (styles 2, 4, and 5), we also consider layered drawings (styles 1 and 3).
Conceptually, our framework for layered drawings of undirected graphs resembles the Sugiyama method for layered drawings of directed graphs \cite{battista_graph_1998}, but there are significant differences. 

In our definition of layering, any edge of $G$ has span at most one, whereas the Sugiyama method only allows edges with span at least one (forbidding edges within the same layer).
Edges with both endpoints on the same layer are called \emph{intralayer edges}, while edges with endpoints on different layers are called \emph{interlayer edges}. \emph{Intralayer crossings} are crossings between two intralayer edges, \emph{interlayer crossings} are those between two interlayer edges, and \emph{mixed crossings} are those between one intralayer edge and one interlayer edge \cite{bachmaier_crossing_2010}.
In typical layered drawings of directed graphs (i.e. Sugiyama-style), all edges are interlayer and drawn as straight line segments \cite{battista_graph_1998}. Bachmaier, Buchner, Forster, and Hong~\cite{bachmaier_crossing_2010} study drawing extended level graphs, which extend layered drawings to allow intralayer edges, drawn as arcs above the horizontal lines. They show that the one-sided crossing minimization problem for extended level graphs is NP-hard, and they provide heuristics for minimizing edge crossings. 

We present a general framework that can find layered drawings of undirected graphs using dynamic programming
formulations that are based on characterizing the interactions between consecutive layers, since
all of our drawing styles rely on layers. For any given layer $i$, we upper bound the number of states it can have. With this upper bound, we enumerate all the possible ways in which two consecutive layers interact. For every possible state of layer $i$, we can obtain a partial drawing with the minimum number of crossings up to this layer. At the last layer, we find the optimum. More details on this dynamic programming framework will be explained later.

\paragraph*{Unnecessary Crossings}  \label{paragraph:crossings}
To avoid unnecessary  crossings, we add these edge constraints:
\begin{enumerate}
    \item Each layer is drawn as an arc diagram, where the arcs are quarter circles.
    \item Interlayer edges are drawn as straight line segments that form an angle of at least 45 degrees with respect to the horizontal direction. 
    \item Arcs  with both endpoints in one layer do not cross arcs  with both endpoints in another layer.
\end{enumerate}
These constraints are part of our drawing model and always achievable as shown below. Thus, if a straight line segment between two layers share an endpoint with an arc in those layers, they will not cross twice, and the shared endpoint is not considered a crossing. See \cref{fig:5-styles}.
These apply to layered drawings and style 5. For style 5, ``layer'' in these constraints would refer to the top and bottom lines, not layers of the input graph. See  \cref{fig:5-styles}.
Constraint 2 and 3 can always be satisfied by increasing the distances between layers. For each layer~$i$, suppose $l_i$ (resp.,~$r_i$) is the $x$-coordinate of the leftmost (resp., rightmost) vertex in layer $i$. If the vertical distance between layers $i$ and $i-1$ is at least $\max \{|r_{i-1} - l_i|, |l_{i-1} - r_i|\}$, then condition 2 will hold for interlayer edges between them.
If the vertical distance between layers $i$ and $i-1$ is at least $ (\sqrt{2} - 1) \max \{r_{i-1} - l_{i-1},   r_i  - l_i \}$, then the longest possible arc below  layer $i-1$ will not cross  the longest possible arc above  layer $i$, and condition 3 will hold.
Thus, the only  crossings are those forced by the topology. We call this style a ``layered drawing.'' 

We can convert such a layered drawing into layer-respecting arc diagrams. On a horizontal line, arrange the vertices from left to right in order of their layer numbers. For any positive integer $k$, the vertices of layer $2k - 1$ are placed in their original order (left-to-right), followed by the vertices of layer $2k$ placed in reverse order (right-to-left). Figure \ref{fig:5-styles} shows this connection between layered drawings and layer-respecting arc diagrams. Note that whether arcs are drawn above or below the horizontal line does not affect the layer-respecting nature of the resulting arc diagram. This observation is only illustrative and not used in our proofs.

\section{Our Framework for the Different Drawing Styles}
%\mike{Please separate the results into two subsections--one for arc diagrams and one for cylindric layouts. Don't worry about the page count limit for now. We can fix that later. Right now reviewers will think the paper is too short and does not have enough results, so we need to space out the results better. And remember to allow for the layering to be given as part of the input}
In this section, we provide algorithms to minimize the number of edge crossings according to the five drawing styles. 
Our framework for achieving our results is based on a simple observation: Given a graph $G=(V,E)$, and a layering, any edge of $G$ has span at most one by our definition of layering; hence,
the number of new crossings introduced by layer $i$ only depends on the previous layer $i-1$, which allows us to process $G$ layer by layer to minimize edge crossings.

Thus, in order for us to implement our framework in each of the five drawing styles, 
it is sufficient for us to consider the possible states for any given layer and characterize the interactions
that can occur between consecutive layers based on the particular drawing style being considered.
We perform this analysis in the following subsections.

\subsection{Arc Diagrams}
In an arc diagram with edges drawn above the horizontal line (style 2), the number of crossings is completely determined by the positions of the vertices, described by the bijective function $\sigma$  \cite{agrawal_eliminating_2024}. 
% To make this more precise, we introduce a definition in \cite{agrawal_eliminating_2024}.
Without loss of generality, we let $\sigma(u) < \sigma(v)$ for any (undirected) edge $(u, v)$. Then for two edges $e = (u, v)$ and $e' = (u', v')$  drawn on the same side of the horizontal line, we can say that $e$ and $e'$ \emph{cross} with respect to $\sigma$ if their endpoints interleave, i.e.,
\(
    \sigma(u) < \sigma(u') < \sigma(v) < \sigma(v') \ \text{or} \ \sigma(u') < \sigma(u) < \sigma(v') < \sigma(v).
\)

\begin{lemma}[Agrawal \textit{et al.}~\cite{agrawal_eliminating_2024}]\label{lemma:cross}
In an arc diagram, two edges can be drawn on the same side of the horizontal line without crossings if and only if they do not cross with respect to $\sigma$.
\end{lemma}

\begin{theorem}\label{thm:arc}
    If we have a graph with a layering of width $w$, then we can find a layer-respecting arc diagram with a minimum number of crossings for the given layering in time:
\begin{enumerate}
    \item $T_w \in O(w^4 (w!)^2 n)$ using style 2;
    \item $T_w \in O(w^4 (w!)^2 2^{3 w^2}n)$ using style 4.
    
\end{enumerate}
\end{theorem}

\subsection{Layered Drawings}
Similarly, crossings in a layered drawing can be determined via the positions of vertices. 
Let the $x$-coordinates of the vertices be specified by a function $\tau : V \to \mathbb{R}$. This is no longer a bijection because distinct vertices in different layers can be vertically aligned.
For intralayer crossings, the crossing condition is the same as in Lemma \ref{lemma:cross} but with $\tau$ replacing $\sigma$. For any interlayer edge $(u, v)$, we will enforce that $u$ is the vertex at the top and $v$ is at the bottom.
We have the following condition for interlayer crossings:
\begin{lemma}[Interlayer crossings]\label{lemma:inter}
    Two interlayer edges $e = (u, v)$ and $e' = (u', v')$ between layers $i$ and $i-1$ can be drawn as line segments without crossings if and only if
\(
    (\tau(u) - \tau(u') ) (\tau(v) - \tau(v')) \ge 0
\).
\end{lemma}
For an interlayer edge  $e = (u, v)$  drawn between layers $i$ and $i-1$, and an intralayer edge $e' = (u', v')$ drawn above layer $i$ (or below layer $i-1$), we say that $e$ and $e'$ \emph{intersect} if $v$ (or $u$) lies in the open line segment between $u'$ and $v'$.
\begin{lemma}[Mixed crossings]\label{lemma:mixed}
For an interlayer edge  $e$  drawn between layers $i$ and $i-1$, and an intralayer edge $e'$ drawn above layer $i$ or below layer $i-1$,  edges $e$ and $e'$ can be drawn without crossings if and only if they do not intersect.
\end{lemma}

\begin{theorem}\label{thm:layered}
    If we have a graph with a layering of width $w$, then we can find a layered drawing with a minimum number of crossings for the given layering in time:
\begin{enumerate}
    \item $T_w \in O(w^4 (w!)^2 n)$ using style 1;
    \item $T_w \in O(w^4 (w!)^2 2^{3 w^2}n)$ using style 3.
\end{enumerate}
\end{theorem}
% \begin{proof}
% The proof is similar to the proof of Theorem \ref{thm:arc}.
% The minimum and maximum $x$-coordinates of the vertices can be chosen first. With these values fixed, we set the distances between layers to be sufficiently large. After this point, the exact values of the  $x$-coordinates no longer affect the crossings, and we only need to keep track of the relative positions. Thus, by Lemma \ref{lemma:inter} and Lemma \ref{lemma:mixed}  the states in  the proof of Theorem \ref{thm:arc} also fully determine the number of crossings in styles 1 and 3, and it takes constant time to check if a pair of edges cross each other. The rest follows.
% \end{proof}

\subsection{Linear Cylindric Drawings}

\begin{theorem} \label{thm:cylindric}
    If we have a graph with a layering of width $w$, then we can find a layer-respecting linear cylindric drawing with a minimum number of crossings for the given layering in time $T_w \in O(w^4 (w!)^2 4^{3 w^2}n)$ using style 5.
\end{theorem}
% \begin{proof}
%     In style 5, recall that the vertices are duplicated, and the top and bottom lines contain the same vertices in the same order (see Figure \ref{fig:5-styles}). Therefore, once we choose the permutation of vertices for a given layer, each edge can be drawn in one of four possible ways: an arc below the top line or above the bottom line, or a straight line segment going up or down.
% These states fully determine the number of crossings.
% Thus, this case is similar to the proof for style 4 in Theorem \ref{thm:arc}.
% \end{proof}

\subsection{BFS Layering}
We now consider the case where no layering with small width is given for the input graph but the graph has BFS width $b$. In this case, we can find a layering by picking a starting vertex and performing BFS. By definition, for any choice of the starting vertex, we have width $w \le b $. Thus, we can apply our previous results to minimize the number of crossings over all possible starting vertices. More precisely, we have the following:
%\mike{Please give the complete proofs of the following corollaries even if they are easy.}
%\mike{Can't we take the minimum choice here and optimize the value of $b$?}
%\alfred{after the generalization i'm not sure if we can find the drawing with minimum crossings of a graph over all its possible layerings. we can enumerate all the BFS layerings to find one that leads to a minimum number of crossings among those. But it's possible that the drawing with the minimum number of crossings has more than one vertices in the first layer.}
\begin{theorem}   \label{cor:time}
        If we have a graph with BFS width $b$, then among its BFS layer-respecting drawings, we can find one with the minimum number of crossings in time:
\begin{enumerate}
    \item $T^* \in O(b^4 (b!)^2 n^2)$ using style 1 or 2;
    \item $T^* \in O(b^4 (b!)^2 2^{3 b^2}n^2)$ using style 3 or 4;
    \item $T^* \in O(b^4 (b!)^2 4^{3 b^2}n^2)$ using style 5.
\end{enumerate}
\end{theorem}

% \begin{proof}
%     For an arbitrary starting vertex $v$, we can compute its BFS layering which will have width at most $b$. By Lemma \ref{lemma:bfs}, the time complexity of BFS is $O(b n)$. 
%  This term will be small compared to the rest of the terms as long as $b$ is not close to 1.
%     \begin{enumerate}
%         \item By Theorems \ref{thm:arc} and \ref{thm:layered}, we can find a drawing in styles 1 or 2 rooted at $v$ with the minimum number of crossings in time $O(b^4 (b!)^2 n)$. We then take the drawing with the minimum number of crossings over all vertices $v$ of the graph, taking $O(b^4 (b!)^2 n^2)$ time total.
%         \item Similarly, by Theorems \ref{thm:arc} and \ref{thm:layered}, we can find a drawing in styles 3 or 4 rooted at $v$ with the minimum number of crossings in time $O(b^4 (b!)^2 2^{3 b^2}n)$. Taking the minimum over all vertices will then take  $O(b^4 (b!)^2 2^{3 b^2}n^2)$ time total.
%         \item Once again, by Theorem \ref{thm:cylindric}, we can find a drawing in style 5 rooted at $v$ with the minimum number of crossings in time $ O(b^4 (b!)^2 4^{3 b^2}n)$. Taking the minimum over all vertices will then take $ O(b^4 (b!)^2 4^{3 b^2}n^2)$ time total.
%     \end{enumerate}
% \end{proof}

\section{Conclusion}
We have presented a number of ways to achieve crossing minimization in fixed-parameter linear time for layer-respecting arc diagrams and layer-respecting linear cylindric drawings, where the parameter is the width of a layered graph.
There are a number of interesting open problems that remain.
For example, one possible direction for future work could be wrap-around linear drawings, such as on a torus rather than a cylinder. Another might be to go beyond the linear cylindric drawing paradigm by allowing edges that (while still drawn as monotonic curves) cross through the line of the vertices, perhaps counting an additional crossing whenever this happens.

% \ifFull
% \subsection*{Acknowledgment}
% We would like to thank David Eppstein for several helpful discussions that led to the results of this paper.
% \fi
\ifAnon
\else
\section*{Acknowledgements}
This research was supported in part by NSF grant CCF-2212129.
\fi
\small
\bibliographystyle{abbrv}
\bibliography{references, zotero-no-url}

@InProceedings{eppstein_et_al:LIPIcs.ESA,
  author =	{Eppstein, David and Goodrich, Michael T. and Liu, Songyu (Alfred)},
  title =	{Bandwidth vs {BFS} Width in Matrix Reordering, Graph Reconstruction, and Graph Drawing},
  booktitle =	{33rd Annual European Symposium on Algorithms (ESA 2025)},
  pages =	{69:1--69:17},
  series =	{Leibniz International Proceedings in Informatics (LIPIcs)},
  ISBN =	{978-3-95977-395-9},
  ISSN =	{1868-8969},
  year =	{2025},
  volume =	{351},
  editor =	{Benoit, Anne and Kaplan, Haim and Wild, Sebastian and Herman, Grzegorz},
  publisher =	{Schloss Dagstuhl -- Leibniz-Zentrum f{\"u}r Informatik},

  URN =		{urn:nbn:de:0030-drops-245373},
  doi =		{10.4230/LIPIcs.ESA.2025.69}
}

@article{burch2021dynamic,
  title={Dynamic Graph Exploration by Interactively Linked Node-Link Diagrams and Matrix Visualizations},
  author={Burch, Michael and Ten Brinke, Kiet Bennema and Castella, Adrien and Peters, Ghassen Karray Sebastiaan and Shteriyanov, Vasil and Vlasvinkel, Rinse},
  journal={Visual Computing for Industry, Biomedicine, and Art},
  volume={4},
  pages={1--14},
  year={2021},
  doi={10.1186/s42492-021-00088-8}
}

@INPROCEEDINGS{debiasi,
  author={Debiasi, Alberto and Simoes, Bruno and De Amicis, Raffaele},
  booktitle={International Conference on Cyberworlds (CW)}, 
  title={Schematization of Node-Link Diagrams and Drawing Techniques for Geo-referenced Networks}, 
  year={2015},
  pages={34-41},
  doi={10.1109/CW.2015.68}}

@article{vlsi,
author = {Chung, Fan R. K. and Leighton, Frank Thomson and Rosenberg, Arnold L.},
title = {Embedding Graphs in Books: A Layout Problem with Applications to {VLSI} Design},
journal = {SIAM Journal on Algebraic Discrete Methods},
volume = {8},
number = {1},
pages = {33-58},
year = {1987},
doi = {10.1137/0608002},
URL = {https://doi.org/10.1137/0608002},
}

@ARTICLE{sugiyama,
  author={Sugiyama, Kozo and Tagawa, Shojiro and Toda, Mitsuhiko},
  journal={IEEE Transactions on Systems, Man, and Cybernetics}, 
  title={Methods for Visual Understanding of Hierarchical System Structures}, 
  year={1981},
  volume={11},
  number={2},
  pages={109-125},
  doi={10.1109/TSMC.1981.4308636}}

@article{lipton1979separator,
  title={A Separator Theorem for Planar Graphs},
  author={Lipton, Richard J and Tarjan, Robert Endre},
  journal={SIAM Journal on Applied Mathematics},
  volume={36},
  number={2},
  pages={177--189},
  year={1979},
  publisher={SIAM}
}

@INPROCEEDINGS{layers,
  author={Dujmovi{\'c}, Vida and Morin, Pat and Wood, David R.},
  booktitle={54th IEEE Symposium on Foundations of Computer Science (FOCS)}, 
  title={Layered Separators for Queue Layouts, {3D} Graph Drawing and Nonrepetitive Coloring}, 
  year={2013},
  pages={280-289},
  doi={10.1109/FOCS.2013.38}}

@article{bannister2019track,
  title={Track Layouts, Layered Path Decompositions, and Leveled Planarity},
  author={Bannister, Michael J and Devanny, William E and Dujmovi{\'c}, Vida and Eppstein, David and Wood, David R},
  journal={Algorithmica},
  volume={81},
  number={4},
  pages={1561--1583},
  year={2019},
  publisher={Springer},
  doi={10.1007/s00453-018-0487-5}
}

@article{eppstein2007confluent,
  title={Confluent Layered Drawings},
  author={Eppstein, David and Goodrich, Michael T and Meng, Jeremy Yu},
  journal={Algorithmica},
  volume={47},
  number={4},
  pages={439--452},
  year={2007},
  publisher={Springer}
}

@inproceedings{agrawal_eliminating_2024,
  title = {Eliminating Crossings in Ordered Graphs},
  booktitle = {19th {{Scandinavian Symposium}} and {{Workshops}} on {{Algorithm Theory}} ({{SWAT}} 2024)},
  author = {Agrawal, Akanksha and Cabello, Sergio and Kaufmann, Michael and Saurabh, Saket and Sharma, Roohani and Uno, Yushi and Wolff, Alexander},
  editor = {Bodlaender, Hans L.},
  year = 2024,
  series = {Leibniz {{International Proceedings}} in {{Informatics}} ({{LIPIcs}})},
  volume = {294},
  pages = {1:1--1:19},
  publisher = {Schloss Dagstuhl -- Leibniz-Zentrum f\"ur Informatik},
  issn = {1868-8969},
  doi = {10.4230/LIPIcs.SWAT.2024.1},
  urldate = {2025-05-09},
  isbn = {978-3-95977-318-8}
}

@inproceedings{auer_plane_2011,
  title = {Plane Drawings of Queue and Deque Graphs},
  booktitle = {Graph {{Drawing}}},
  author = {Auer, Christopher and Bachmaier, Christian and Brandenburg, Franz Josef and Brunner, Wolfgang and Glei{\ss}ner, Andreas},
  editor = {Brandes, Ulrik and Cornelsen, Sabine},
  year = 2011,
  pages = {68--79},
  publisher = {Springer},
  doi = {10.1007/978-3-642-18469-7_7},
  isbn = {978-3-642-18469-7},
  langid = {english}
}

@article{bachmaier_crossing_2010,
  title = {Crossing Minimization in Extended Level Drawings of Graphs},
  author = {Bachmaier, Christian and Buchner, Hedi and Forster, Michael and Hong, Seok-Hee},
  year = 2010,
  journal = {Discrete Applied Mathematics},
  volume = {158},
  number = {3},
  pages = {159--179},
  issn = {0166-218X},
  doi = {10.1016/j.dam.2009.09.002},
  urldate = {2025-05-09}
}

@article{bannister_crossing_2018,
  title = {Crossing {{Minimization}} for 1-Page and 2-Page {{Drawings}} of {{Graphs}} with {{Bounded Treewidth}}},
  author = {Bannister, Michael and Eppstein, David},
  year = 2018,
  journal = {Journal of Graph Algorithms and Applications},
  volume = {22},
  number = {4},
  pages = {577--606},
  issn = {1526-1719},
  doi = {10.7155/jgaa.00479},
  urldate = {2026-04-25},
  langid = {english}
}

@book{battista_graph_1998,
  title = {Graph {{Drawing}}: {{Algorithms}} for the {{Visualization}} of {{Graphs}}},
  shorttitle = {Graph {{Drawing}}},
  author = {Battista, Giuseppe Di and Eades, Peter and Tamassia, Roberto and Tollis, Ioannis G.},
  year = 1998,
  edition = {1st},
  publisher = {Prentice Hall PTR},
  isbn = {978-0-13-301615-4}
}

@misc{bekos_online_2023,
  title = {An Online Framework to Interact and Efficiently Compute Linear Layouts of Graphs},
  author = {Bekos, Michael A. and Haug, Mirco and Kaufmann, Michael and M{\"a}nnecke, Julia},
  year = 2023,
  number = {arXiv:2003.09642},
  eprint = {2003.09642},
  primaryclass = {cs},
  publisher = {arXiv},
  doi = {10.48550/arXiv.2003.09642},
  urldate = {2025-10-05},
  archiveprefix = {arXiv}
}

@article{bernhart_book_1979,
  title = {The Book Thickness of a Graph},
  author = {Bernhart, Frank and Kainen, Paul C},
  year = 1979,
  journal = {Journal of Combinatorial Theory, Series B},
  volume = {27},
  number = {3},
  pages = {320--331},
  issn = {0095-8956},
  doi = {10.1016/0095-8956(79)90021-2},
  urldate = {2025-05-24}
}

@inproceedings{chaplick_monotone_2024,
  title = {Monotone Arc Diagrams with Few Biarcs},
  booktitle = {32nd {{International Symposium}} on {{Graph Drawing}} and {{Network Visualization}} ({{GD}} 2024)},
  author = {Chaplick, Steven and F{\"o}rster, Henry and Hoffmann, Michael and Kaufmann, Michael},
  editor = {Felsner, Stefan and Klein, Karsten},
  year = 2024,
  series = {Leibniz {{International Proceedings}} in {{Informatics}} ({{LIPIcs}})},
  volume = {320},
  pages = {11:1--11:16},
  publisher = {Schloss Dagstuhl -- Leibniz-Zentrum f\"ur Informatik},
  issn = {1868-8969},
  doi = {10.4230/LIPIcs.GD.2024.11},
  urldate = {2025-05-18},
  isbn = {978-3-95977-343-0}
}

@article{dubey_hardness_2011,
  title = {Hardness Results for Approximating the Bandwidth},
  author = {Dubey, Chandan and Feige, Uriel and Unger, Walter},
  year = 2011,
  journal = {Journal of Computer and System Sciences},
  series = {Celebrating {{Karp}}'s {{Kyoto Prize}}},
  volume = {77},
  number = {1},
  pages = {62--90},
  issn = {0022-0000},
  doi = {10.1016/j.jcss.2010.06.006},
  urldate = {2025-03-30}
}

@inproceedings{dunagan_euclidean_2001,
  title = {On Euclidean Embeddings and Bandwidth Minimization},
  booktitle = {Approximation, {{Randomization}}, and {{Combinatorial Optimization}}: {{Algorithms}} and {{Techniques}}},
  author = {Dunagan, John and Vempala, Santosh},
  editor = {Goemans, Michel and Jansen, Klaus and Rolim, Jos{\'e} D. P. and Trevisan, Luca},
  year = 2001,
  pages = {229--240},
  publisher = {Springer},
  doi = {10.1007/3-540-44666-4_26},
  isbn = {978-3-540-44666-8},
  langid = {english}
}

@inproceedings{healy_how_2002,
  title = {How to {{Layer}} a {{Directed Acyclic Graph}}},
  booktitle = {Graph {{Drawing}}},
  author = {Healy, Patrick and Nikolov, Nikola S.},
  editor = {Mutzel, Petra and J{\"u}nger, Michael and Leipert, Sebastian},
  year = 2002,
  pages = {16--30},
  publisher = {Springer},
  doi = {10.1007/3-540-45848-4_2},
  isbn = {978-3-540-45848-7},
  langid = {english}
}

@misc{hegemann_chunky_2026,
  title = {Chunky {{Chains}}: {{Graph Drawings}} on {{Small Screens}}},
  shorttitle = {Chunky {{Chains}}},
  author = {Hegemann, Tim and Jilg, Dominik and Sieper, Marie Diana and Wolf, Samuel},
  year = 2026,
  number = {arXiv:2607.06029},
  eprint = {2607.06029},
  primaryclass = {cs.DS},
  publisher = {arXiv},
  doi = {10.48550/arXiv.2607.06029},
  urldate = {2026-08-15},
  archiveprefix = {arXiv}
}

@inproceedings{klawitter_experimental_2018,
  title = {Experimental Evaluation of Book Drawing Algorithms},
  booktitle = {Graph {{Drawing}} and {{Network Visualization}}},
  author = {Klawitter, Jonathan and Mchedlidze, Tamara and N{\"o}llenburg, Martin},
  editor = {Frati, Fabrizio and Ma, Kwan-Liu},
  year = 2018,
  pages = {224--238},
  publisher = {Springer International Publishing},
  doi = {10.1007/978-3-319-73915-1_19},
  isbn = {978-3-319-73915-1},
  langid = {english}
}

@inproceedings{komarek_network_2015,
  title = {Network {{Visualization Survey}}},
  booktitle = {Computational {{Collective Intelligence}}},
  author = {Komarek, Ales and Pavlik, Jakub and Sobeslav, Vladimir},
  editor = {N{\'u}{\~n}ez, Manuel and Nguyen, Ngoc Thanh and Camacho, David and Trawi{\'n}ski, Bogdan},
  year = 2015,
  pages = {275--284},
  publisher = {Springer International Publishing},
  doi = {10.1007/978-3-319-24306-1_27},
  isbn = {978-3-319-24306-1},
  langid = {english}
}

@article{masuda_crossing_1990,
  title = {Crossing Minimization in Linear Embeddings of Graphs},
  author = {Masuda, S. and Nakajima, K. and Kashiwabara, T. and Fujisawa, T.},
  year = 1990,
  journal = {IEEE Transactions on Computers},
  volume = {39},
  number = {1},
  pages = {124--127},
  issn = {1557-9956},
  doi = {10.1109/12.46286},
  urldate = {2025-05-09}
}

@inproceedings{ruegg_generalization_2016,
  title = {A {{Generalization}} of the {{Directed Graph Layering Problem}}},
  booktitle = {Graph {{Drawing}} and {{Network Visualization}}},
  author = {R{\"u}egg, Ulf and Ehlers, Thorsten and Sp{\"o}nemann, Miro and {von Hanxleden}, Reinhard},
  editor = {Hu, Yifan and N{\"o}llenburg, Martin},
  year = 2016,
  pages = {196--208},
  publisher = {Springer International Publishing},
  doi = {10.1007/978-3-319-50106-2_16},
  isbn = {978-3-319-50106-2},
  langid = {english}
}

@article{shahrokhi_book_1996,
  title = {The Book Crossing Number of a Graph},
  author = {Shahrokhi, Farhad and Sz{\'e}kely, L{\'a}szl{\'o} A. and S{\'y}kora, Ondrej and Vrt'o, Imrich},
  year = 1996,
  journal = {Journal of Graph Theory},
  volume = {21},
  number = {4},
  pages = {413--424},
  issn = {1097-0118},
  doi = {10.1002/(SICI)1097-0118(199604)21:4<413::AID-JGT7>3.0.CO;2-S},
  urldate = {2025-05-18},
  langid = {english}
}

@inproceedings{wattenberg_arc_2002,
  title = {Arc Diagrams: Visualizing Structure in Strings},
  booktitle = {{{IEEE}} Symposium on Information Visualization, 2002. {{INFOVIS}} 2002.},
  author = {Wattenberg, M.},
  year = 2002,
  pages = {110--116},
  doi = {10.1109/INFVIS.2002.1173155}
}

\clearpage
\section*{Appendix}
% \appendix
% \section{Omitted Proofs}
\subsection{Arc Diagrams}

\begin{restate}{thm:arc}
    If we have a graph with a layering of width $w$, then we can find a layer-respecting arc diagram with a minimum number of crossings for the given layering in time:
\begin{enumerate}
    \item $T_w = O(w^4 (w!)^2 n)$ using style 2;
    \item $T_w = O(w^4 (w!)^2 2^{3 w^2}n)$ using style 4.
    
\end{enumerate}
\end{restate}
\begin{proof}

\begin{figure}[hbt]
\vspace*{-12pt}
    \centering
    \includegraphics[width=.9\linewidth]{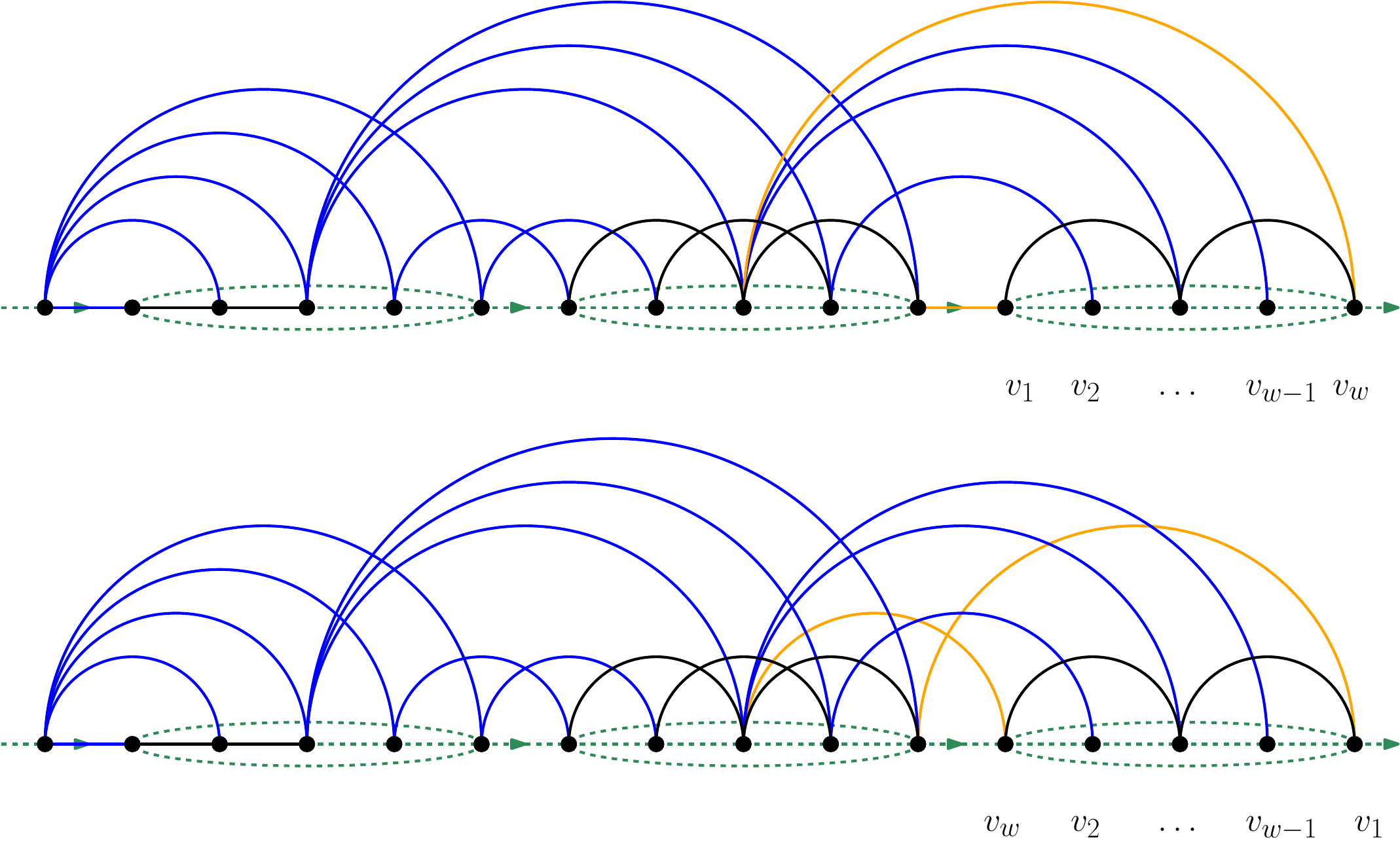}
    \caption{Permutation of vertices within the rightmost layer determines the number of crossings newly introduced when adding this layer. For example, when we swap $v_1$ and $v_w$, the edges that connect them and vertices in the previous layer will lead to different crossings.
    If we were allowed to draw edges below the horizontal line, it could potentially reduce the number of crossings.}
    \label{fig:thm1proof}
\end{figure}
We first explain the high-level strategy.
Suppose that layer $i$ is in state $s_i$ and the total number of possible states for any given layer is at most $S_w$. The exact state will be specified later. Let $c(s_i)$ denote the minimum number of crossings of the partial drawing involving only layer $i$ in state $s_i$ and all earlier layers. At the first layer, we have our base cases. Let $c(s_1)$ denote  the number of crossings introduced by adding layer $1$ in state $s_1$.
At layer $i$, we can ignore the crossings between layers $i$ and $i+1$. Let $c(s_{i-1}, s_i)$ denote the number of crossings introduced by adding layer $i$ in state $s_i$ when layer $i-1$ is in state $s_{i-1}$. 
    Each new crossing is between two edges, involving four vertices, and at least one of these vertices is in layer $i$. We set up the recurrence as follows:
\begin{align*}
\forall s_i \text{ where } i > 1, \ c(s_i) = \min_{s_{i-1}} (c(s_{i-1}) + c(s_{i-1}, s_i)).
\end{align*}

The global minimum number of crossings for drawings with the given layering is $\min_{s_{H}} c(s_{H})$ where $H = O(n)$ is the number of layers. A drawing that achieves this can be found by backtracking to the state which achieved the minimum for the recurrence at each step.
We need to compute $c(s_{i-1}, s_i)$ for every possible combination of  $s_{i-1}$ and $ s_i$ for each $i > 1$ and this dominates the time complexity.
Suppose that we can compute $c(s_{i-1}, s_i)$ in time at most $C_w$ for each $i > 1$. There are $S_w^2$ combinations per step; hence the total time complexity is $T_w = \Theta(C_w S_w^2 H)$.

Next, we specify the possible states needed for any given layer $i$. In style 2, the states are permutations of vertices within layer~$i$. (See \cref{fig:thm1proof}.) 
In style 4, arcs can be drawn either above or below the line. 
The states now need to account for these two possibilities for each new edge  introduced by adding layer $i$. There are at most $\binom{w}{2}$ edges within layer $i$, and at most $w^2$ edges between layer $i-1$ and layer $i$.
Then for each drawing style, we have the following upper bounds on $S_w$:
\begin{enumerate}
    \item If we follow style 2, then $S_w = w!$. 
    \item If we follow style 4, then $S_w \le w! 2^{\binom{w}{2} + w^2} \le w! 2^{\frac{3}{2} w^2} $. 
    % \item We follow style 5. $S_w \le w! 4^{\binom{w}{2} + w^2} \le w! 4^{\frac{3}{2} w^2} $. 
\end{enumerate}

    For the worst-case analysis of $C_w$, suppose that each layer has $w$ vertices. There are at most $n_w \coloneq \binom{w}{2} + w^2$ edges introduced by adding layer $i$. Each pair of these edges may cross, giving $\binom{n_w}{2}$ possible pairs. Furthermore, every new edge may cross a previous edge introduced by adding layer $i - 1$, adding another $n_w^2$ possible pairs. For each pair of edges, the states specify whether they are on the same side of the line. If they are, by \cref{lemma:cross}, we can decide if they cross each other in constant time. 
    In the worst case, we can compute $c(s_{i-1}, s_i)$ in time 
\begin{align*}
    C_w \coloneq \binom{n_w}{2} + n_w^2 = \Theta(w^4).
\end{align*}
We plug these upper bounds into $T_w$ to complete the proof.
\end{proof}

\subsection{Layered Drawings}
\begin{restate}{thm:layered}
    If we have a graph with a layering of width $w$, then we can find a layered drawing with a minimum number of crossings for the given layering in time:
\begin{enumerate}
    \item $T_w \in O(w^4 (w!)^2 n)$ using style 1;
    \item $T_w \in O(w^4 (w!)^2 2^{3 w^2}n)$ using style 3.
\end{enumerate}
\end{restate}
\begin{proof}
The proof is similar to the proof of Theorem \ref{thm:arc}.
The minimum and maximum $x$-coordinates of the vertices can be chosen first. With these values fixed, we set the distances between layers to be sufficiently large. After this point, the exact values of the  $x$-coordinates no longer affect the crossings, and we only need to keep track of the relative positions. Thus, by \cref{lemma:cross,lemma:inter,lemma:mixed}  the states in  the proof of Theorem \ref{thm:arc} also fully determine the number of crossings in styles 1 and 3, and it takes constant time to check if a pair of edges cross each other. The rest follows.
\end{proof}

\subsection{Linear Cylindric Drawings}

\begin{restate} {thm:cylindric}
    If we have a graph with a layering of width $w$, then we can find a layer-respecting linear cylindric drawing with a minimum number of crossings for the given layering in time $T_w \in O(w^4 (w!)^2 4^{3 w^2}n)$ using style 5.
\end{restate}
\begin{proof}
In style 5, recall that the vertices are duplicated, and the top and bottom lines contain the same vertices in the same order (see Figure \ref{fig:5-styles}). Style 5 can be visualized as a layered drawing with two large layers. The crossing conditions are \cref{lemma:cross,lemma:inter,lemma:mixed}.   Therefore, once we choose the permutation of vertices for a given layer of the input graph, each edge can be drawn in one of four possible ways: an arc below the top line or above the bottom line, or a straight line segment going up or down.
These states fully determine the number of crossings.
Thus, this case is similar to the proof for style 4 in Theorem \ref{thm:arc}.
\end{proof}

\subsection{BFS Layering}
\begin{restate}   {cor:time}
        If we have a graph with BFS width $b$, then among its BFS layer-respecting drawings, we can find one with the minimum number of crossings in time:
\begin{enumerate}
    \item $T^* \in O(b^4 (b!)^2 n^2)$ using style 1 or 2;
    \item $T^* \in O(b^4 (b!)^2 2^{3 b^2}n^2)$ using style 3 or 4;
    \item $T^* \in O(b^4 (b!)^2 4^{3 b^2}n^2)$ using style 5.
\end{enumerate}
\end{restate}

\begin{proof}
    For an arbitrary starting vertex $v$, we can compute its BFS layering which will have width at most $b$. By Lemma \ref{lemma:bfs}, the time complexity of BFS is $O(b n)$. 
 This term will be small compared to the rest of the terms as long as $b$ is not close to 1.
    \begin{enumerate}
        \item By Theorems \ref{thm:arc} and \ref{thm:layered}, we can find a drawing in styles 1 or 2 rooted at $v$ with the minimum number of crossings in time $O(b^4 (b!)^2 n)$. We then take the drawing with the minimum number of crossings over all vertices $v$ of the graph, taking $O(b^4 (b!)^2 n^2)$ time total.
        \item Similarly, by Theorems \ref{thm:arc} and \ref{thm:layered}, we can find a drawing in styles 3 or 4 rooted at $v$ with the minimum number of crossings in time $O(b^4 (b!)^2 2^{3 b^2}n)$. Taking the minimum over all vertices will then take  $O(b^4 (b!)^2 2^{3 b^2}n^2)$ time total.
        \item Once again, by Theorem \ref{thm:cylindric}, we can find a drawing in style 5 rooted at $v$ with the minimum number of crossings in time $ O(b^4 (b!)^2 4^{3 b^2}n)$. Taking the minimum over all vertices will then take $ O(b^4 (b!)^2 4^{3 b^2}n^2)$ time total.
    \end{enumerate}
\end{proof}
\end{document}